\documentclass[letterpaper, 10 pt, conference]{ieeeconf}

\makeatletter
\newcommand\semihuge{\@setfontsize\semihuge{22.3}{22}}
\makeatother

\usepackage{tabularx}

\usepackage{algpseudocode}
\usepackage{algorithm}
\usepackage{algorithmicx}

\usepackage{lipsum} 
\usepackage{arydshln}
\usepackage[dvips]{color}
\usepackage{comment}
\usepackage{todonotes}
\usepackage{epsf}
\usepackage{epsfig}
\usepackage{times}
\usepackage{epsfig}
\usepackage{graphicx}
\usepackage{bbold}
\usepackage{mathtools}
\usepackage{mathrsfs}
\usepackage{amssymb}
\usepackage{pdfpages}
\usepackage{epstopdf}
\newfloat{algorithm}{t}{lop}

\usepackage{amsmath}
\usepackage{dsfont}
\usepackage{lettrine} 

\usepackage{amsmath,epsfig,amssymb,algorithm,algpseudocode,amsthm,cite,url}
\usepackage{subcaption}
\allowdisplaybreaks
\usepackage{csquotes}

\usepackage{verbatim}
\usepackage[english]{babel}
\usepackage{cuted}
\usepackage[keeplastbox]{flushend}
\usepackage{amsmath,amssymb}
\DeclareMathOperator{\E}{\mathbb{E}}
\DeclareMathOperator*{\argmax}{arg\,max}

\usepackage[justification=centering]{caption}
\usepackage{verbatim}

\newtheorem{theorem}{\bf Theorem}

\newtheorem{proposition}{\bf Proposition}

\begin{document}
	%
		\IEEEoverridecommandlockouts
	\title{Robust Deep Reinforcement Learning for Security and Safety in Autonomous Vehicle Systems}
	\author{Aidin Ferdowsi$^{1}$, Ursula Challita$^{2}$, Walid Saad$^{1}$, and Narayan B. Mandayam$^3$  
		\thanks{This research was supported by the U.S. National Science Foundation under Grants OAC-1541105 and OAC-1541069.}
		\thanks{$^{1}$Aidin Ferdowsi and Walid Saad are with Wireless@VT, Bradley Department of Electrical and Computer Engineering, Virginia Tech, Blacksburg, VA, USA,
			{\tt\small \{aidin, walids\}@vt.edu}}%
		\thanks{$^{2}$Ursula Challita is with School of Informatics, The University of Edinburgh, Edinburgh, UK,
			{\tt\small ursula.challita@ed.ac.uk}}%
		\thanks{$^{3}$Narayan B. Mandayam is with WINLAB, Dept. of ECE, Rutgers University, New Brunswick, NJ, USA,
			{\tt\small narayan@winlab.rutgers.edu}}%
		\vspace*{-100cm}}
	\makeatletter
	\patchcmd{\@maketitle}
	{\addvspace{0.5\baselineskip}\egroup}
	{\addvspace{-5\baselineskip}\egroup}
	{}
	{}
	\makeatother
	\maketitle
	\thispagestyle{empty}
	\pagestyle{empty}

	%
	\IEEEpeerreviewmaketitle
	
\begin{abstract}	
To operate effectively in tomorrow's smart cities, autonomous vehicles (AVs) must rely on intra-vehicle sensors such as camera and radar as well as inter-vehicle communication. Such dependence on sensors and communication links exposes AVs to cyber-physical (CP) attacks by adversaries that seek to take control of the AVs by manipulating their data. Thus, to ensure safe and optimal AV dynamics control, the data processing functions at AVs must be robust to such CP attacks. To this end, in this paper, the state estimation process for monitoring AV dynamics, in presence of CP attacks, is analyzed and a novel adversarial deep reinforcement learning (RL) algorithm is proposed to maximize the robustness of AV dynamics control to CP attacks. The attacker's action and the AV's reaction to CP attacks are studied in a game-theoretic framework. In the formulated game, the attacker seeks to inject faulty data to AV sensor readings so as to manipulate the inter-vehicle optimal safe spacing and potentially increase the risk of AV accidents or reduce the vehicle flow on the roads. Meanwhile, the AV, acting as a defender, seeks to minimize the deviations of spacing so as to ensure robustness to the attacker's actions. Since the AV has no information about the attacker's action and due to the infinite possibilities for data value manipulations, the outcome of the players' past interactions are fed to long-short term memory (LSTM) blocks. Each player's LSTM block learns the expected spacing deviation resulting from its own action and feeds it to its RL algorithm. Then, the the attacker's RL algorithm chooses the action which maximizes the spacing deviation, while the AV's RL algorithm tries to find the optimal action that minimizes such deviation. Simulation results show that the proposed adversarial deep RL algorithm can improve the robustness of the AV dynamics control as it minimizes the intra-AV spacing deviation.
\end{abstract}
\section{Introduction}
Intelligent transportation systems (ITS) will encompass autonomous vehicles (AVs), roadside smart sensors (RSSs), vehicular communications, and even drones \cite{ferdowsi2017deep,Mozaffari2016,zeng2018joint}. To operate in a truly autonomous manner in future ITSs, AVs must be able to process a large volume of ITS data collected via a plethora of sensors and communication links. Reliability of such data is crucial to mitigate the chances of vehicle collisions and improve the vehicular flow on the roads \cite{Amoozadeh2015}. However, this reliance on communications and data processing renders AVs highly susceptible to cyber-physical attacks. In particular, an attacker can possibly interject the AV data processing stage, reduce the reliability of measurements by injecting faulty data, and ultimately induce accidents or compromise the traffic flow in the ITS\cite{Kargl2008}. Such flow disruptions can also cascade to other interdependent critical infrastructure such as power grids or cellular communication systems that provide service to the ITSs \cite{ferdowsi2017colonel,Ferdowsi2017CB}.

Recently, a number of security solutions have been proposed for addressing intra-vehicle security problems \cite{Kleberger2011,Woo2015, Chaudhry2012, Studnia2013}. In \cite{Kleberger2011}, the authors identified the key vulnerabilities of a vehicle's controller and proposed a number of intrusion detection algorithms to secure this controller. Moreover in \cite{Woo2015}, the authors show that long-range wireless attacks on AVs' current security protocols can disrupt their controller area network. They analyze the vulnerabilities of AVs' intra-vehicle networks to outsider wireless attack. Meanwhile, the authors in \cite{Chaudhry2012} addressed the security challenges of plug-in electric vehicles, while accounting for their impact on the power system. Moreover, a survey on security threats and protection mechanisms in embedded automotive networks is presented in \cite{Studnia2013}.  

Furthermore, vehicular communication security challenges and solutions have also been studied recently in \cite{Calandriello2011,Kim2010,Sun2017,PETRILLO2018,Tuohy2015}. In \cite{Calandriello2011}, the security vulnerabilities of current vehicular communication architecture are analyzed. Moreover, the computational overhead caused by beacon encryption is mitigated by using a short term authentication scheme and a cooperative vehicle computation architecture. The authors in \cite{Kim2010} proposed the use of multi-source filters to reduce the security vulnerability of a vehicular network, with respect to data injection attacks. Furthermore, the in \cite{Sun2017} introduced a new framework to improve the trustworthiness of beacons by combining two physical measurements (angle of arrival and Doppler effect) from received wireless signals. Moreover, in \cite{PETRILLO2018}, the authors proposed a collaborative control strategy for vehicular platooning to address spoofing and denial of service attacks. Finally, an overview of current research on advanced
intra-vehicle networks and the smart components of ITS and their applications is presented in \cite{Tuohy2015}.

However, the architecture and solutions in \cite{Kleberger2011,Woo2015, Chaudhry2012, Studnia2013,Calandriello2011,Kim2010,Sun2017,PETRILLO2018,Tuohy2015} do not take into account the interdependence between the cyber and physical layers of AVs while designing their security solutions. Moreover, these existing works do not properly model the attacker's actions and goals. In this context, the cyber-physical interdependence of the attacker's actions and goals will help providing better security solutions. Moreover, the prior art in \cite{Kleberger2011,Woo2015, Chaudhry2012, Studnia2013,Calandriello2011,Kim2010,Sun2017,PETRILLO2018,Tuohy2015}, does not provide solutions that can enhance the robustness of AV dynamics control to attacks. Nevertheless, designing an optimal and safe ITS requires robustness to attacks on intra-vehicle sensors as well as inter-vehicle communication. Moreover, existing works on ITS security often assume a stable state for the attacker's action, while in many practical scenarios, the attacker might adaptively change its strategy to increase the impact of its attack on the ITS.

The main contribution of this paper is, thus, to propose a novel adversarial deep reinforcement learning (RL) framework that aims at providing robust AV control. In particular, we consider a car following model in which we focus on the control of an AV that closely follows another AV. Such a model is suitable because it captures the AV's dynamics control while taking intro account AV's sensor readings and beaconing. We consider four sources of information about the leading AV gathered from intra-vehicle sensors such as camera, radar, RSSs, and inter-vehicle beaconing. We consider an attacker which can inject bad data to such information and tries to increase the risk of accidents or reduce the vehicular flow. In contrast, the AV's goal is to optimally control its speed while staying robust to such data injection attacks from the attacker. To analyze the interactions between the AV and the attacker, we pose the problem as a game and analyze its Nash equilibrium (NE). However, we observe that obtaining the AV and attacker actions at the NE will be challenging due to having continuous attacker and AV actions sets as well as continuous AV speed and spacing. To address this problem, we propose two deep neural networks (DNNs) based on long-short term memory (LSTM) blocks for the AV and the attacker that extract the summary of past AV dynamics and feed such summaries to an RL algorithm for each player. On the one hand, the AV's RL algorithm tries to learn the best estimation from its leading AV's speed by combining the sensor readings. On the other hand, the RL algorithm for the attacker tries to deceive the AV and deviate the inter-vehicle optimal safe spacing. Simulation results show that the proposed deep RL algorithm converges to a mixed-strategy Nash equilibrium point and can lead to significant improvement in the AV's robustness to data injection attacks. The results also show that the AV can use the proposed deep RL algorithm to effectively learn the sensor fusion rule that minimizes the error in speed estimation thus reducing the deviations from the optimal safe spacing.

The rest of the paper is organized as follows. Section
\ref{sec:Model} introduces the system model for AV control. Section \ref{sec:Problem} formulates the robust AV control problem in a game-theoretic framework.
Section \ref{sec:Learning} proposes our adversarial deep learning algorithm. Section \ref{sec:Sim} analyzes the simulation and conclusions are drawn
in Section \ref{sec:Conc}.
\section{System Model}\label{sec:Model}
Consider a smart road in an ITS consisting of multiple AVs and RSSs. Each AV $ i $ is equipped with a camera to take images from the environment, a radar to measure distances from objects in the vicinity of the AV, and a transceiver device to communicate important position-speed-acceleration (PVA) beacons with nearby AVs and sensors over a cellular network. One challenging area in such ITSs is the optimal and safe flow control of AVs by using the collected measurements and received beacons. Moreover, the presence of an adversary might induce faulty decisions to the ITS and result in accidents or reduce the vehicle flow. Thus, the AVs' control on the roads must be robust to faulty data injected to the measurements and beacons by a malicious attacker. Next, we present an estimation model at each AV to observe the speed of its leading AV (i.e., the preceding AV) using sensor and beacon fusion and we model the adversary and its available actions. Then, we define a dynamic process to capture the spacing between the AVs as a function of the attacker and AV actions.
\subsection{Autonomous Vehicle Cyber-physical System}
In order to drive safely and prevent accidents, each AV $ i $ must acquire information about its own position, speed $ v_i $ as well as the distance and speed of some nearby objects such as the immediately leading AV, $ i-1 $. One framework to analyze the speed of an AV $i$ is by using the so called \emph{car-following} models that is popular in the literature \cite{BRACKSTONE1999181}. Here, we use the General Motors' first car-following model to analyze the speed update at each AV $ i $ as a function of AV $ i-1 $'s speed as follows\cite{BRACKSTONE1999181}:
\begin{align}\label{eq:GM}
	\dot{v}_i(t)=\lambda (\hat{v}_{i-1}(t)-v_{i}(t)),
\end{align}
where $ \lambda $ is a reaction parameter and $ \hat{v}_{i-1}(t) $ is the estimated speed of AV $ i-1 $ at AV $ i $. As we can see from \eqref{eq:GM}, each vehicle must estimate $ \hat{v}_{i-1}(t) $ at each time step in order to control its dynamics. To this end, each AV must use its own built-in sensors such as camera, radar as well as periodic reports from AV $ i-1 $ and the closest RSS. Thus, at each AV $ i $, AV $ i-1 $'s speed, $ v_{i-1}(t) $, must be estimated from AV $ i-1 $'s measured speed using a camera image, $ c_i $, and radar reading, $ r_i $, on AV $ i $, AV $ i-1 $'s speed report $ u_{i-1} $, and closest RSS's reported speed $ s_i $. Therefore, the relationship between AV $ i-1 $'s exact speed and the measurements can be expressed using a \emph{generic linear model} as follows:
\begin{align}
	\boldsymbol{z}_i(t)=\boldsymbol{H}_iv_{i-1}(t)+\boldsymbol{e}_i(t),
\end{align}
where $ \boldsymbol{H} \in \mathbb{R}^{4\times 1} $ is the measurement Jacobian matrix, $ \boldsymbol{z}_i\triangleq\left[c_i,r_i,u_{i-1}, s_i\right]^T $, and $ \boldsymbol{e}_i \in \mathbb{R}^{4\times 1} $ is a random error vector. Now, assuming complete information about $ \boldsymbol{H} $ and with a condition that $ \boldsymbol{H} $ is full rank, we can estimate $ d_i $ as follows:
\begin{align}
	\boldsymbol{H}_i^T\boldsymbol{z}_i(t)&=\boldsymbol{H}_i^T\boldsymbol{H}_id_i(t)+\boldsymbol{H}_i^T\boldsymbol{e}_i(t) \nonumber \\	\Rightarrow\left[\boldsymbol{H}^T_i\boldsymbol{H}_i\right]^{-1}\hspace{-3mm}\boldsymbol{H}_i^T\boldsymbol{z}_i(t)&=\left[\boldsymbol{H}^T_i\boldsymbol{H}_i\right]^{-1}\hspace{-3mm}\boldsymbol{H}_i^T\boldsymbol{H}_id_i(t)\nonumber\\&+\left[\boldsymbol{H}^T_i\boldsymbol{H}_i\right]^{-1}\hspace{-3mm}\boldsymbol{H}_i^T\boldsymbol{e}_i(t)\nonumber\\
	\Rightarrow d_i(t)&=\underbrace{\left[\boldsymbol{H}^T_i\boldsymbol{H}_i\right]^{-1}\hspace{-3mm}\boldsymbol{H}_i^T\boldsymbol{z}_i(t)}_{\hat{d}_i(t)}\nonumber\\
	&-\left[\boldsymbol{H}^T_i\boldsymbol{H}_i\right]^{-1}\hspace{-3mm}\boldsymbol{H}_i^T\boldsymbol{e}_i(t),
\end{align}
where $ \hat{d}_i $ is the estimated distance. Now, by defining $ \hat{\boldsymbol{z}}_i\triangleq \boldsymbol{H}_i\tilde{d}_i $ as the estimated measurement vector, we can find the measurement estimation error or residual as $ \tilde{\boldsymbol{z}}_i\triangleq\boldsymbol{z}-\hat{\boldsymbol{z}}_i $. Next, we can define a weighted cost function for the measurement residual as follows:
\begin{align}\label{eq:cost}
	J_i(\tilde{\boldsymbol{z}}_i)\triangleq\tilde{\boldsymbol{z}}_i^T\boldsymbol{W}_i\tilde{\boldsymbol{z}}_i=\left[\boldsymbol{z}_i-\hat{\boldsymbol{z}}_i\right]^T\boldsymbol{W}_i\left[\boldsymbol{z}_i-\hat{\boldsymbol{z}}_i\right],
\end{align}
where $ \boldsymbol{W}_i $ is a positive definite square matrix. If the measurements are not dependent, a typical choice for $ \boldsymbol{W}_i $ is to have positive diagonal components while the non-diagonal components are zero. Since in our model the sensor error are independent, we consider $\boldsymbol{W}_i $ to be a diagonal matrix in which $ w^i_{k} $ on the $ k $-th row and column of $ \boldsymbol{W}_i $ is the weight of measurement $ k $. The estimator at each AV $ i $ must minimize the cost function in \eqref{eq:cost}. It can be proven that the solution of this problem is given by \cite{dorf2011modern}:
\begin{align}\label{eq:estimation}
	\bar{v}_{i-1}(t)=\left[\boldsymbol{H}_i^T\boldsymbol{W}_i\boldsymbol{H}_i\right]^{-1}\boldsymbol{H}_i^T\boldsymbol{W}_i\boldsymbol{z}_i(t).
\end{align}

Since we know that all the sensors can directly measure the speed, we can consider $ \boldsymbol{H}=\left[1,1,1,1\right]^T $. Moreover, since the diagonal entities of $ \boldsymbol{W}_i $ are weights assigned to each sensor reading, thus we can consider $ \sum_{i=1}^{k}w^i_{k}=1 $. Now, \eqref{eq:estimation} can be simplified to:
\begin{align}\label{eq:estimation2}
	\bar{v}_{i-1}(t)=\frac{\sum_{k=1}^{4}w^i_k(t)z^i_k(t)}{\sum_{k=1}^{4}w^i_k(t)}=\sum_{k=1}^{4}w^i_k(t)z^i_k(t)=\boldsymbol{w}^T_i(t)\boldsymbol{z}_i(t),
\end{align}
where $ z^i_k(t) $ is the $ k $-th element of $ \boldsymbol{z}_i(t) $, and $ \boldsymbol{w}_i(t) $ is a vector with $ w^k_i(t) $ as its element $ k $.
\subsection{Attack Model}
In the studied system, an attacker is able to inject faulty data to any of the aforementioned sensor readings. Such attacks can take place using special lasers to alter camera and radar readings as well as man in the middle attacks to inject bad data into the input of the AV and RSS beacons. We define $ \tilde{\boldsymbol{z}}_i(t) $ as an ``under attack sensor vector" which can be defined as $\tilde{\boldsymbol{z}}^i(t)\triangleq \boldsymbol{z}^i(t)+\boldsymbol{a}^i(t),$ where $ \boldsymbol{a}^i(t) $ is the injected faulty data vector at time $ t $ to the sensor vector $ \boldsymbol{z}_i(t) $. Thus, such attack will induce a deviation in the value of the speed estimation which can be derived from \eqref{eq:estimation2} as follows:
\begin{align}\label{eq:estimationUattack}
	\tilde{v}_{i-1}(t)&=\boldsymbol{w}^T_i(t)\tilde{\boldsymbol{z}}_i(t),\nonumber\\
	&=v_{i-1}(t)+\boldsymbol{w}_i^T(t)\boldsymbol{e}_i(t)+\boldsymbol{w}_i^T(t)\boldsymbol{a}_i(t).
\end{align}
Hence, the attacker can change AV $ i-1 $'s estimated speed at AV $ i $ by injecting faulty data. However, to stay stealth, the attacker cannot inject any arbitrary data due to the physical limitations of the system. For instance, at each time step the attacker cannot report a very high or low speed to AV $ i $. Moreover, due to the difference in the sensor types (camera image, radar reading, beacons), the attacker cannot manipulate the sensors equally. Thus, we consider threshold levels for each sensor $ k $ such that $ |a_k^i(t)| < \tau^i_k $ where $ a_k^i(t)$ is the data injected to AV $ i $'s sensor $ k $.
\section{Cyber-physical Security Problem and Game Formulation}\label{sec:Problem}
From \eqref{eq:estimationUattack}, we can see that the AV $ i $'s estimated speed at each time step is a function of the actual AV $ i $'s speed, $ v_{i-1}(t) $, as well as the noise, $ \boldsymbol{e}_i(t) $, the weighting $ \boldsymbol{w}_i(t) $, and the attack $ \boldsymbol{a}_i(t) $ vectors. Thus, using \eqref{eq:GM} we can see that each AV $ i $'s speed $ v_i(t) $ is also a function of $ v_{i-1}(t) $,  $ \boldsymbol{e}_i(t) $, $ \boldsymbol{w}_i(t) $, and $ \boldsymbol{a}_i(t) $. Here, we analyze the spacing $d_i(t)$ between AVs $ i $ and $ i-1 $, before we subsequently investigate the optimal safe spacing for the AVs. It can be shown that the derivative of $ d_i(t) $ is the difference between the speeds of the AVs, as follows:
\begin{align}\label{eq:spacingspeed}
	\dot{d}_i(t)={v}_{i-1}(t)-v_{i}(t,\boldsymbol{a}_i(t),\boldsymbol{w}_i(t),\boldsymbol{e}_i(t)).
\end{align}

Thus, the spacing at each time step is a function of the AVs' speed as well as the sensor readings and the attack vector. Such attack vector can manipulate the spacing $ d_i(t) $ yielding two effects on the ITS: a) if $ d_i(t) $ decreases, the risk of collision between AVs increases and b) if $ d_i(t) $ increases, the traffic flow will reduce, which will be non-optimal and ineffective for the ITS operation. Therefore, the attacker's goal is to manipulate the spacing and deviate it from the optimal safe state while staying stealthy. In contrast, AV $ i $ tries to optimize its operation while staying robust to such sensor manipulations to minimize the spacing deviations. Formally, the attacker's goal is to find an attack vector $ \boldsymbol{a}_i^*(t) $ at each time step $ t $ such that:
\begin{align}\label{eq:att}
	\boldsymbol{a}_i^*(t)&=\max_{\boldsymbol{a}_i(t)} R_i(\boldsymbol{a}_i(t),\boldsymbol{w}_i(t),\boldsymbol{e}_i(t))\triangleq(d_i(t)-o_i(v_{i-1}(t)))^2\hspace{-1mm},\\
	&\text{s.t. }  |a^i_k(t)|<{\tau}_k^i\,\, \forall k=1,\dots,4
\end{align}
where $ R_i(t) $ is AV $ i $'s \emph{regret} function which quantifies the deviation from the optimal safe spacing and $ o_i(v_{i-1}(t)) $ is the optimal safe spacing at time $ t $. Conversely, AV $ i $'s objective is to find a weighting vector $ \boldsymbol{w}^*_i(t) $ to minimize the defined regret function as follows:
\begin{align}\label{eq:def}
	\boldsymbol{w}_i^*(t)&=\min_{\boldsymbol{w}_i(t)} R_i(\boldsymbol{a}_i(t),\boldsymbol{w}_i(t)),\\
	&\text{s.t. }  \sum_{k=1}^{4} {w}^i_k(t)=1.
\end{align}
The optimization problems in \eqref{eq:att} and \eqref{eq:def} are dependent on the actions of both the attacker and the AV. Solving such problem requires taking into account the interdependence of AV and the attacker's actions. In the following we analyze the interdependence of the attacker and the AV's actions to each other and their previous actions and we formulate such problem in a game-theoretic framework \cite{han2012game}.

To this end, we first derive the impact of past attacker and AV actions on their future actions. Next, we analytically derive a limit on the number of past regret samples which are enough to take future actions with $ T $ being the sampling period of the sensors.
\begin{theorem}\label{theorem:pastsamples}
	The attacker and the AV can optimally choose their future actions if: (i) $ \lambda T <2 $ and (ii) they have information about the regret for at least $ \bar{n} $ past time steps, where $ \bar{n} $ is the smallest integer that satisfies:
	\begin{align}
		\bar{n} \leq \frac{\log (\epsilon)}{\log(|1-\lambda T|)},
	\end{align}
	 where $ \epsilon $ is a small value.
\end{theorem}
\begin{proof}
	 First, due to discrete-time sensor readings, we convert the continuous car-following model in \eqref{eq:GM} to a discrete one while considering that AV $ i-1 $'s estimated speed is under attack, as follows:
	\begin{align}\label{eq:GMdisc}
	\frac{v_i(t+T)-v_i(t)}{T}&=\lambda(\tilde{v}_{i-1}(t)-v_i(t))\nonumber\\
	\Rightarrow	v_i(t+T)&=\lambda T \tilde{v}_{i-1}(t)+(1-\lambda T)v_i(t).
	\end{align} 
	From \eqref{eq:GMdisc} we can see that having a stable system requires $ |1-\lambda T| < 1 \Rightarrow 0<\lambda T <2  $ which always holds true. Moreover, we can use \eqref{eq:GMdisc} to find $ v_i(t) $ as a function of $  v_i(t-T) $ and $  \tilde{v}_{i-1}(t-T) $ as follows:
	\begin{align}
		v_i(t)&=\lambda T \tilde{v}_{i-1}(t-T)+(1-\lambda T)v_i(t-T)
	\end{align}
	and thus we will have:
	\begin{align}
		v_i(t+T)&=\lambda T \tilde{v}_{i-1}(t)\nonumber\\
		&+(1-\lambda T)\Big(\lambda T \tilde{v}_{i-1}(t-T)+(1-\lambda T)v_i(t-T)\Big)
	\end{align}
	By continuing this process for the past steps, we can establish the following relationship between AV $ i $'s future speed, its initial speed, and AV $ i-1 $'s past speed values:
	\begin{align}\label{eq:past}
		v_i(nT)&=(1-\lambda T)^{n-1}v_i(0)\nonumber\\
		&+\sum_{l=0}^{n}\lambda T (1-\lambda T)^l \tilde{v}_{i-1}\left((n-l)T\right).
	\end{align}
	\eqref{eq:past} shows that AV $ i-1 $'s older speed values have smaller effect on AV $ i $'s speed decision. Moreover, for large values of $n$, the term $(1-\lambda T)^{n-1} $ tends to zero. Thus, we can approximate \eqref{eq:past} as follows:
	\begin{align}\label{eq:approximate}
		v_i(nT) = \sum_{l=0}^{\bar{n}}\lambda T (1-\lambda T)^l \tilde{v}_{i-1}\left((n-l)T\right),
	\end{align}
	where $ \bar{n} $ is the smallest integer number which satisfies $ \bar{n} \leq \frac{\log (\epsilon)}{\log(|1-\lambda T|)} $, where $ \epsilon $ is a small value that can be defined based on the maximum allowable speed of an $ AV $ on the road. \eqref{eq:approximate} shows the dependence of $ v_i $ on the past values of $ v_{i-1} $. Now, using \eqref{eq:estimationUattack} we can find also the relationship between AV $ i $'s speed and its own past actions as well as those of the attacker, as follows:
	\begin{align}\label{eq:approximateattack}
	v_i(n) &= \sum_{l=0}^{\bar{n}}\lambda T (1-\lambda T)^l {v}_{i-1}\left(n-l\right)+ \sum_{l=0}^{\bar{n}}\lambda T(1-\lambda T)^l\nonumber\\
	&\times\big(\boldsymbol{w}_i^T(n-l)\boldsymbol{e}_i(n-l)+\boldsymbol{w}_i^T(n-l)\boldsymbol{a}_i(n-l)\big).
	\end{align}
	Thus, at each time step, only the $ \bar{n} $ past actions of the attacker and AV $ i $ will affect AV $ i $'s future speed if $ \lambda T < 2 $. Note that, in \eqref{eq:approximateattack}, we dropped $ T $ from the arguments of all the time variant functions for notational simplicity.
\end{proof}
Theorem \ref{theorem:pastsamples} proves that in order to solve the optimization problems in \eqref{eq:att} and \eqref{eq:def}, the AV and the attacker can only use their past $ \hat{n} $ actions.
Next, we derive the initial spacing between AVs when AV $ i $ for the first time must start following $ i-1 $ to converge to an optimal safe spacing, which will be useful in the solution of our problem.
\begin{proposition}\label{proposition:initial}
	The spacing between AVs $ i $ and $ i-1 $ converges to an optimal safe spacing if AV $ i-1 $ must start following AV $ i-1 $ when the spacing is:
	\begin{align}\label{eq:initialspacing}
		d^*(\nu)\hspace{-1mm}= \hspace{-1mm}o(\nu) - (\hat{n}+2) T \nu +T\nu \sum_{p=0}^{\hat{n}-1}\left(1- (1-\lambda T)^p\right),
	\end{align}
	where $ d^*(\nu) $ is the spacing when the AV $ i $ starts following AV $ i $, and $ \nu $ is the expectation of $ v_{i-1} $, $ \E\{v_{i-1}\} = \nu $.
\end{proposition}
\begin{proof}
	From \eqref{eq:spacingspeed}, we can find the following:
	\begin{align}
		d_i(n+1)&=d_i(n)+T(v_{i-1}(n+1)-v_i(n+1))\nonumber\\
		&=d_i(n)+Tv_{i-1}(n+1)\nonumber\\
		&-\sum_{l=0}^{\bar{n}}\lambda T^2 (1-\lambda T)^l \tilde{v}_{i-1}\left(n-l+1\right).
	\end{align}
	Let the estimation $ \tilde{v} $ be noise-free and attack-free. Then, we can find $ d_i $ as a function of its initial state and only AV $ i $'s speed as follows:
	\begin{align}
		d_i(n+1) &= d_i(0) + \sum_{p=0}^{n}Tv_{i-1}(p+1)\nonumber\\
		&-\hspace{-1mm}\sum_{p=0}^{n}\hspace{-2mm}\sum_{l=0}^{\min\{\bar{n},p\}}\lambda T^2 (1-\lambda T)^l v_{i-1}\left(p-l+1\right),
	\end{align}
	 Then, the expectation of the spacing will be:
	\begin{align}
		\E\{d_i(n+1)\}&=d_i(0)+nT\nu - \nu \sum_{p=0}^{n}\hspace{-1mm}\sum_{l=0}^{\min\{\bar{n},p\}}\hspace{-3mm}\lambda T^2 (1-\lambda T)^l \nonumber \\
		&=d_i(0)+nT\nu\nonumber\\
		&-\nu\sum_{p=0}^{n}\lambda T^2\frac{1-(1-\lambda T)^{\min\{\bar{n},p\}}}{1-(1-\lambda T)}\nonumber\\
		&=d_i(0)+nT\nu - T\nu \sum_{p=0}^{\hat{n}-1}\left(1- (1-\lambda T)^p\right)\nonumber\\
		& -  T\nu \sum_{p=\hat{n}}^{n}\left(1- (1-\lambda T)^{p}\right)\nonumber\\
		& =d_i(0)+nT\nu - (n-(\hat{n}-1)+1) T \nu\nonumber\\
		&-T\nu \sum_{p=0}^{\hat{n}-1}\left(1- (1-\lambda T)^p\right)\nonumber\\
		&=d_i(0) \hspace{-1mm} + \hspace{-1mm}(\hat{n}+2) T \nu \hspace{-1mm}-\hspace{-1mm}T\nu\hspace{-1mm} \sum_{p=0}^{\hat{n}-1}\hspace{-1mm}\left(1- (1-\lambda T)^p\right).\nonumber
	\end{align}
	Given that our goal is to reach an optimal safe spacing, we have $ \E\{d_i(n+1)\} = o(\nu) $, and, then, we can derive \eqref{eq:initialspacing}.
\end{proof}
Proposition \ref{proposition:initial} shows that AV $ i $ must start following AV $ i-1 $ when it reaches a distance of $ d^*(\nu) $ from AV $ i-1 $ while AV $ i-1 $'s speed is $ \nu $. Now, if $ d_i(0)=d_i^*(\nu) $, we can formally define the regret function of \eqref{eq:att} as follows:
\begin{align}
	R_i(n)\hspace{-1mm}&= \hspace{-1mm}\lambda^2 T^4\hspace{-0.5mm} \Big[\hspace{-0.5mm}\sum_{p=0}^{n}\hspace{-1.5mm}\sum_{l=0}^{\min\{\bar{n},p\}}\hspace{-3.5mm} (1-\lambda T)^l \Big( \boldsymbol{w}_i^T(p-l+1)\boldsymbol{e}_i(p-l+1)\nonumber\\
	&+\boldsymbol{w}_i^T(p-l+1)\boldsymbol{a}_i(p-l+1)\Big)\Big]^2,\vspace{4mm}
\end{align}
where for notational simplicity we use $ R_i(n) $ instead of $ R_i(\boldsymbol{a}_i(n),\boldsymbol{w}_i(n),\boldsymbol{e}_i(n)) $ as defined in \eqref{eq:att}.
We can see that, at each time step $ n $, the regret function accumulates the errors from the initial time step till $ n $. Thus, if we define the \emph{deviation} from optimal safe spacing at time step $ n $ as follows:

\begin{align}\label{eq:deviation}
	\delta_i(n)&\triangleq\sum_{p=0}^{n}\hspace{-1mm}\sum_{l=0}^{\min\{\bar{n},p\}}\hspace{-3mm} (1-\lambda T)^l \Big( \boldsymbol{w}_i^T(p-l+1)\boldsymbol{e}_i(p-l+1)\nonumber\\&+\boldsymbol{w}_i^T(p-l+1)\boldsymbol{a}_i(p-l+1)\Big),
\end{align}
then, we can derive the deviation as a process as follows:

\begin{strip}\vspace{-2mm}
	\begin{align} \label{eq:deviationprocess}
	\delta_i(n)\hspace{-1mm}=\hspace{-1mm}\delta_i(n-1)\hspace{-1mm}+\hspace{-4.5mm} \underbrace{\sum_{l=0}^{\min\{\bar{n},n\}}\hspace{-3mm} (1-\lambda T)^l \hspace{-0.5mm}\Big(\hspace{-0.5mm} \boldsymbol{w}_i^T(\min\{\bar{n},n\}-\hspace{-0.5mm}l\hspace{-0.5mm}+1\hspace{-0.5mm})\boldsymbol{e}_i(\min\{\bar{n},n\}-\hspace{-0.5mm}l\hspace{-0.5mm}+1\hspace{-0.5mm})+\boldsymbol{w}_i^T(\min\{\bar{n},n\}-\hspace{-0.5mm}l\hspace{-0.5mm}+1\hspace{-0.5mm})\boldsymbol{a}_i(\min\{\bar{n},n\}-\hspace{-0.5mm}l\hspace{-0.5mm}+1\hspace{-0.5mm})\hspace{-0.5mm}\Big)}_{\theta\left(\boldsymbol{w}_i,\boldsymbol{a}_i,\boldsymbol{e}_i\right)}\hspace{-0.5mm}.
	\end{align}\vspace{-5mm}
\end{strip}

Thus, we can write the regret function as follows:
\begin{align}
	R_i(n)=\lambda^2 T^4 \Big[\delta_i(n-1)+\theta\left(\boldsymbol{w}_i(n),\boldsymbol{a}_i(n),\boldsymbol{e}_i(n)\right)\Big]^2.
\end{align}
Then, at each time step $ n $, the attacker and the AV must choose their associated vectors using their past $ \hat{n} $ actions and the deviation from last step, $ \delta_i(n-1) $.

We now formally define a noncooperative game where the players are the attacker and the AV, the AV's action $ \alpha^{\text{AV}}(n) $ is to choose a weighting vector at each time step, $ \boldsymbol{w}_i(n) $, and the attacker's action, $ \alpha^{\text{att}}(n) $ is to choose a data injection vector at each time step, $ \boldsymbol{a}_i(n) $. Moreover, the AV's utility function is $ U^{\text{AV}}(n) = -R_i(n) $ while the attacker's utility function is $ U^{\text{att}} = R_i(n) $. A suitable solution concept for the defined game is the so-called \emph{Nash equilibrium (NE)} which a stable game state at which the AV cannot reduce the regret by unilaterally changing its action $ \boldsymbol{w}_i(n) $ given that the action of the attacker is fixed. Moreover, at the NE, the attacker cannot increase the regret by changing its action $ \boldsymbol{a}_i(n) $ while the AV keeps its action fixed. Since the players utility at each time step sum up to zero, the game is zero-sum and is guaranteed to admit at least one \emph{mixed-strategy Nash equilibrium (MSNE)}\cite{bacsar1998dynamic}. A \emph{mixed strategy} is a randomization between the available actions of the AV which satisfy $ \sum_{i=1}^nw_i^k =1 $ and the available actions of the attacker which satisfy $ |a_i^k|<\tau_k,\,\, \forall k =1,\dots,4 $. Even though the MSNE exists for our game, it is analytically challenging to derive the equilibrium strategies. Thus, we next propose a deep RL algorithm for this game in which the AV and the attacker learn their optimal actions based on their time-varying observations of each others' actions.

\section{Adversarial Deep Reinforcement Learning for Optimal Safe AV Control}\label{sec:Learning}
The proposed deep RL algorithm have two components: (i) A DNN that summarizes the past actions and spacing deviations and (ii) an RL component, which can be used by each player to decide on the best action to choose based on the summary from the DNN, as shown in Fig. \ref{fig:DRL}.

To derive the AV and the attacker's actions that maximize their expected utility using RL, we use a Q-learning algorithm \cite{heinrich2016deep}. In this algorithm, we define a state-action value $ Q $-function $ Q^j(s^j,\alpha^j) $ which is the expected return of player $ j $ when starting at a \emph{state} $ s^j $ and performing action $ \alpha^j $. To derive the maximizer action at each time step for each player, we use the following update rule for the Q function \cite{heinrich2016deep}:
\begin{align}\label{eq:Qfunction}
	&Q^j_{n+1}(s^j(n),\alpha^j(n))\hspace{-0.5mm}=\nonumber\\\hspace{-0.5mm}	&Q^j_n(s^j(n),\alpha^j(n))+\beta \Big[U^j(n+1)+\hspace{-0.7mm}\nonumber\\
	&\gamma\hspace{-0.7mm}\max_{\alpha^j}\hspace{-0.5mm}Q^j_{n+1}(\hspace{-0.5mm}s^j\hspace{-0.5mm}(n+1),\hspace{-0.5mm}\alpha^j\hspace{-0.5mm})-\hspace{-0.5mm}Q^j_{n}\hspace{-0.5mm}(\hspace{-0.5mm}s^j(n),\alpha^j(n)\hspace{-0.5mm})\hspace{-0.5mm}\Big]\hspace{-0.5mm},
\end{align}
where $ \beta $ is the learning rate and $ \gamma $ is the discount factor. In our problem, $ \alpha^{\text{att}} = \boldsymbol{a}_i(n) $ is the attacker's action while, $ \alpha^{\text{AV}} = \boldsymbol{w}_i(n) $ is the AV's action. Moreover, since the players have no information about the other player's past actions and noise vector, the observed state for the AV is $ s^{\text{AV}}(n) =\left\{\hspace{-0.5mm}\boldsymbol{w}_i(n-\hat{n})\hspace{-0.5mm}, \hspace{-0.5mm}\dots\hspace{-0.5mm}, \hspace{-0.5mm} \boldsymbol{w}_i(n-1)\hspace{-0.5mm};\hspace{-0.5mm}\delta_i(n-\hat{n})\hspace{-0.5mm},\hspace{-0.5mm}\dots\hspace{-0.5mm},\hspace{-0.5mm}\delta_i(n-1)\hspace{-0.5mm}\right\} \hspace{-0.5mm}$ while the observed state for the attacker is $ s^{\text{att}}(n) =\left\{\boldsymbol{a}_i(n-\hat{n}), \dots, \boldsymbol{a}_i(n-1);\delta_i(n-\hat{n}),\dots,\delta_i(n-1)\right\} $. From \eqref{eq:Qfunction}, we can see that at each step, the players must find the action which maximizes  $ Q^j_n $. However, to be able to find such action, each player must know all the possible states, and at each time step find the maximizer action over the observed action. However, in our problem, all the available states cannot be stored since, $ \boldsymbol{w}_i(n) $, $ \boldsymbol{a}_i(n) $, $ \boldsymbol{e}_i(n) $, and $ \delta_i(n) $ have continuous values which will result in an infinite state space.

\begin{figure*}
	\centering
	\includegraphics[width=0.8\textwidth]{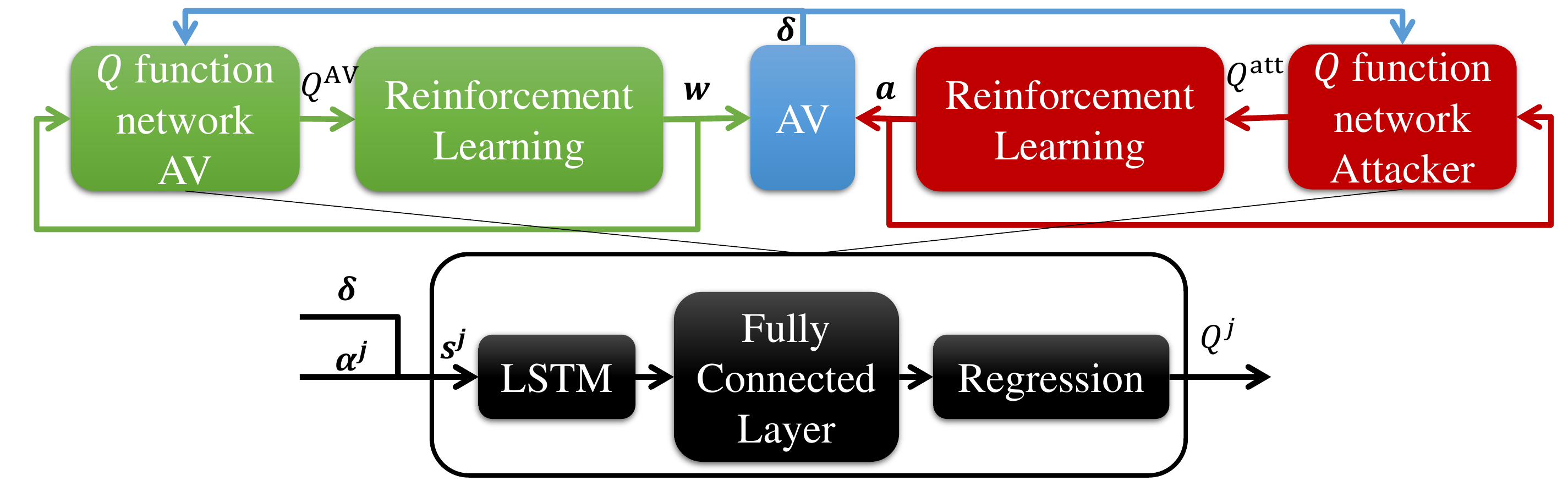}
	\caption{The architecture for the proposed adversarial deep RL algorithm.}
	\label{fig:DRL}
\end{figure*}

To solve such a challenging problem, we use DNNs which are very effective at extracting features from large data sets. Particularly, we use long short term memory (LSTM) blocks which are deep recurrent neural networks (RNNs) that can store information for long periods of time and, thus, can learn long-term dependencies within
a given sequence \cite{ferdowsi2018deep,chen2017machine,LSTM}. Essentially, an LSTM algorithm processes an input sequence $ s^j(n) $ by adding new information into a memory, and using gates which control the extent to which new information should be memorized, old information should be forgotten, and current information should be used. Therefore, the output of an LSTM algorithm will be impacted by the network activation in previous time steps. Thus, LSTMs are suitable for our problem in which
we want to extract useful features from actions and deviation of previous time steps and reduce our state space. Thus, the proposed deep RL algorithm will use a DNN as shown in Fig. \ref{fig:DRL} to approximate the $ Q $ function for each player and using this $ Q $-function we will choose optimal actions for each player from \eqref{eq:Qfunction}. Algorithm \ref{Algorithm:DeepRL} summarizes the proposed adversarial deep RL approach that is used by each player to learn its optimal action vectors. Moreover, Fig. \ref{fig:DRL} shows the DNN architecture for the proposed adversarial deep RL algorithm. Using the proposed algorithm, we can find the optimal actions for the players and it will converge to one of the MSNE points of the game \cite{heinrich2016deep}.

\begin{algorithm}[t]
	\caption{Adversarial Deep RL for Robust AV Control}
	\begin{algorithmic}[1]\footnotesize 
		\State Initialize two \emph{replay memory} $ M^j $ that stores the past experiences of the players and two DNNs for $ Q^j $.  
		\State Observe initial state $ s^j(0) $ for both players.
		\State \textbf{Repeat:}
		\State \quad Select an action $ \alpha^j $ for each player $ j $:
		\State \quad \quad with probability $ \varepsilon $ select a random action,
		\State \quad \quad otherwise select $ \alpha^j = \argmax_{\alpha'^j} Q^j(s^j(n),\alpha'^j) $.
		\State \quad Perform action $ \alpha^j $ for both players simultaneously. 
		\State \quad Observe utility $ U^j(n+1) $ and new state $ s^j(n+1) $.
		\State \quad Store \emph{experience} $ \left\{s^j(n),\alpha^j(n),U^j(n+1),s^j(n+1)\right\} $ in replay \Statex \quad memory $ D^j $ for each player $ j $.
		\State \quad Sample a random experience $ \left\{\hspace{-0.5mm}\hat{s}^j(\eta), \hspace{-0.5mm}\hat{\alpha}^j(\eta),\hspace{-0.5mm}\hat{U}^j(\eta+1),\hspace{-0.5mm}\hat{s}^j(\eta+1)\hspace{-0.5mm}\right\} $
		\Statex \quad from the replay memory $ D^j $ for each player.
		\State \quad Calculate the \emph{target} value $ t^j $ for each player $ j $:
		\State \quad \quad If the sampled experience is for $ n=0 $ then $ t^j=\hat{U}^j  $,
		\State \quad \quad Otherwise $ t^j=\hat{U}^j + \gamma \max_{\alpha'^j} Q^j(\hat{s}^j(n+1),\alpha'^j) $.
		\State \quad Train the network $ Q^j $ for each player using:
		\Statex \quad $ [t^j-Q^j(\hat{s}^j(n),\hat{a}^j(n)) ]^2$.
		\State \quad $ n = n+1 $.
		\State \textbf{Until} convergence to an MSNE
	\end{algorithmic}
	\label{Algorithm:DeepRL}
\end{algorithm}

\section{Simulation Results and Analysis}\label{sec:Sim}
For our simulations, we choose a reaction parameter $ \lambda = 1 $ and a sampling period $ T=1 $. Using Theorem \ref{theorem:pastsamples}, by choosing $ \epsilon = 0.001 $ we find $ \hat{n} = 66 $ which is equivalent to $ 6.6 $ seconds. This means that each AV only needs the information about the past $ 6.6 $ seconds to be able to carry out an optimal safe action. Moreover, we consider that the sensor noise powers are arranged in a descending order as follows: RSS, radar, camera and beacon. This is due to the fact that the RSS might have the highest error for speed measurement while the beacon is sending the exact speed information from AV $ i-1 $ to AV $ i $. In addition, we do not supply the information about noise statistics to the AV and the attacker. Thus, they both must learn such information during the interaction with each other. Moreover, we consider that the attack threshold levels for the sensors are $ \tau_1 = 0.5$ m/s , $ \tau_1 = 1$ m/s , $tau_3 =1$ m/s, and $  \tau_1 = 1.5  $ m/s.

\begin{figure}[!t]
	\centering
	\includegraphics[width=0.95\columnwidth]{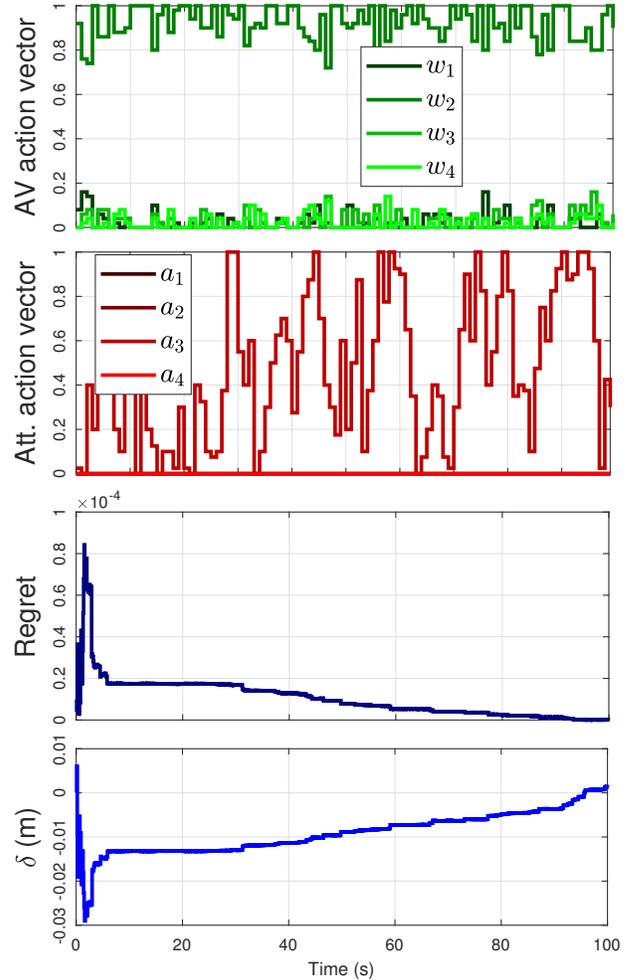}
	\caption{The AV and the attacker's action, regret, and deviation for our proposed algorithm in the case where the attacker attacks only to the beacon information.}
	\label{fig:att2beacon}
\end{figure}

In our first simulation, we consider a case in which the attacker can only attack beacon values as it is one of the most studied attacks in the literature. Also, since the beacon has the lowest error power, the ideal case for the AV is to put the highest weight on beacon information in the absence of the attacker. Fig. \ref{fig:att2beacon} shows the action vector for the AV and the attacker when they interact for the first 100 seconds during which AV $ i $ follows $ i-1 $. From Fig. \ref{fig:att2beacon} we can see that, even though the beacon has the lowest error power, since the attacker attacks the beacons, the AV decides to put more weight on other sensors. Moreover, Fig. \ref{fig:att2beacon} shows that, although the attacker can always have a data injection that is equal to the threshold level, $ \tau_3 = 1 $, it can sometimes decide to inject lower values, to maximize the expected deviation. In addition, Fig. \ref{fig:att2beacon} shows that, in the first steps of the learning procedure, the attacker can inject deviations in the spacing thus increasing the regret for the AV. However, our proposed deep RL algorithm enables the AV to mitigate the error on the estimation and thus stay robust to the data injected attack. Therefore, after 100 seconds, we see that the regret reduces to zero and the attacker cannot force the AV to deviate from its optimal safe spacing. 

\begin{figure}[!t]
	\centering
	\includegraphics[width=0.95\columnwidth]{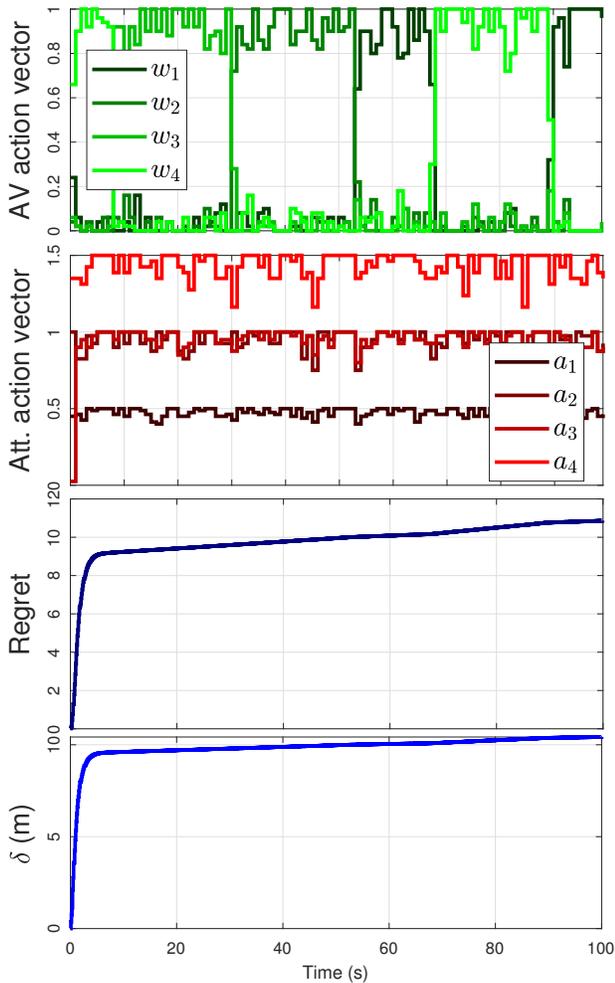}
	\caption{The AV and the attacker's action, regret, and deviation for our proposed algorithm in the case where the attacker attacks all the sensors.}
	\label{fig:att2all}
\end{figure}
\begin{figure}[!t]
	\centering
	\includegraphics[width=\columnwidth]{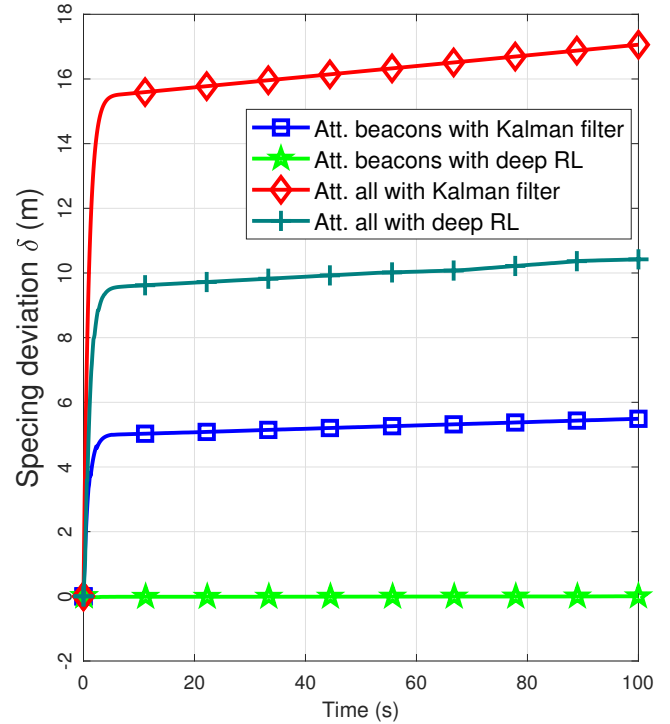}
	\caption{Comparison of the proposed deep RL algorithm with a baseline that does not use any learning process.}
	\label{fig:comp}
\end{figure}

Next, in another simulation, we consider the worst case security scenario, where the attacker can attack all of the sensor readings. Fig. \ref{fig:att2all} shows the AV and the attacker action process during the first 100 seconds of car-following. In this case, we can see from Fig. \ref{fig:att2all} that the attacker can attack to all the sensor values and thus, the AV cannot prioritize between the sensor readings as in the previous simulation. Thus, Fig. \ref{fig:att2all} shows that the AV tries to assign higher weights to one step in small time periods to deceive the attacker. In contrast, the attacker tries to maximize the value of injected data as seen from Fig. \ref{fig:att2all} that the injected data values are close to the threshold level. Moreover, Fig. \ref{fig:att2all} shows that in the first 10 seconds the value of regret has an abrupt increase, while in the remaining time the regret stays almost constant. Also, the spacing deviation reaches a value close to $ 10 $ meter. This means that, when the attacker can attack all of the sensor values, the AV cannot make the estimation robust to the injected attack, however the regret stays approximately constant. Thus, the AV can feedback the spacing deviation $ \delta $ to its car following model to compensate the deviation from the optimal safe spacing by changing the speed and thus make the AV resilient to such data injection attacks.

In Fig. \ref{fig:comp} we show the spacing deviation as a function of time. In this figure, we compare our proposed deep RL algorithm with a baseline scenario, where the AV knows the noise distributions and choose a static weighting vector $ \boldsymbol{w}_i $ using a \emph{Kalman} filter. Fig. \ref{fig:comp} shows that, even though the used Kalman filter converges to a constant spacing deviation, however, our proposed deep RL algorithm has a lower steady state deviation than the Kalman filter. This is due to the fact that the Kalman filter only takes into account the noise power, however, our proposed algorithm uses an adversarial approach to learn the attacker's action. This, indeed, enables the AV to minimize the deviation from the optimal safe spacing and remain more robust to the attacker.

\section{Conclusion}\label{sec:Conc}
In this paper, we have proposed a novel deep RL method which enables a robust dynamics control for AVs in presence of data injection attacks on their sensor readings. To analyze the incentives of attacker to attack on the AV data and address the AV's reaction to such attacks, we have formulated a game-theoretic problem between the attacker and the AV. We have shown that, deriving the mixed strategies at Nash equilibrium is analytically challenging. Thus, we have used our proposed deep RL algorithm to learn the optimal sensor fusion for the AV  at each time step that results in minimizing the deviation from an optimal safe inter-vehicle spacing. In the proposed deep RL algorithm, we have used LSTM blocks which can extract temporal features and dependence of AV and attacker actions and deviation values and feed them to a reinforcement learning algorithm. Simulation results show that, using the proposed deep RL algorithm, an AV can mitigate the effect of data injection attacks on the sensor data and thus stay robust to such attacks.
\bibliographystyle{IEEEtran}
\bibliography{references}

\begin{thebibliography}{10}
\providecommand{\url}[1]{#1}
\csname url@samestyle\endcsname
\providecommand{\newblock}{\relax}
\providecommand{\bibinfo}[2]{#2}
\providecommand{\BIBentrySTDinterwordspacing}{\spaceskip=0pt\relax}
\providecommand{\BIBentryALTinterwordstretchfactor}{4}
\providecommand{\BIBentryALTinterwordspacing}{\spaceskip=\fontdimen2\font plus
\BIBentryALTinterwordstretchfactor\fontdimen3\font minus
  \fontdimen4\font\relax}
\providecommand{\BIBforeignlanguage}[2]{{%
\expandafter\ifx\csname l@#1\endcsname\relax
\typeout{** WARNING: IEEEtran.bst: No hyphenation pattern has been}%
\typeout{** loaded for the language `#1'. Using the pattern for}%
\typeout{** the default language instead.}%
\else
\language=\csname l@#1\endcsname
\fi
#2}}
\providecommand{\BIBdecl}{\relax}
\BIBdecl

\bibitem{ferdowsi2017deep}
A.~Ferdowsi, U.~Challita, and W.~Saad, ``Deep learning for reliable mobile edge
  analytics in intelligent transportation systems,'' \emph{arXiv preprint
  arXiv:1712.04135}, 2017.

\bibitem{Mozaffari2016}
M.~Mozaffari, W.~Saad, M.~Bennis, and M.~Debbah, ``Unmanned aerial vehicle with
  underlaid device-to-device communications: Performance and tradeoffs,''
  \emph{IEEE Transactions on Wireless Communications}, vol.~15, no.~6, pp.
  3949--3963, June 2016.

\bibitem{zeng2018joint}
T.~Zeng, O.~Semiari, W.~Saad, and M.~Bennis, ``Joint communication and control
  for wireless autonomous vehicular platoon systems,'' \emph{arXiv preprint
  arXiv:1804.05290}, 2018.

\bibitem{Amoozadeh2015}
M.~Amoozadeh, A.~Raghuramu, C.~n.~Chuah, D.~Ghosal, H.~M. Zhang, J.~Rowe, and
  K.~Levitt, ``Security vulnerabilities of connected vehicle streams and their
  impact on cooperative driving,'' \emph{IEEE Communications Magazine},
  vol.~53, no.~6, pp. 126--132, June 2015.

\bibitem{Kargl2008}
F.~Kargl, P.~Papadimitratos, L.~Buttyan, M.~Müter, E.~Schoch, B.~Wiedersheim,
  T.~V. Thong, G.~Calandriello, A.~Held, A.~Kung, and J.~P. Hubaux, ``Secure
  vehicular communication systems: implementation, performance, and research
  challenges,'' \emph{IEEE Communications Magazine}, vol.~46, no.~11, pp.
  110--118, November 2008.

\bibitem{ferdowsi2017colonel}
A.~Ferdowsi, W.~Saad, and N.~B. Mandayam, ``{Colonel Blotto} game for secure
  state estimation in interdependent critical infrastructure,'' \emph{arXiv
  preprint arXiv:1709.09768}, 2017.

\bibitem{Ferdowsi2017CB}
A.~Ferdowsi, W.~Saad, B.~Maham, and N.~B. Mandayam, ``A {Colonel Blotto} game
  for interdependence-aware cyber-physical systems security in smart cities,''
  in \emph{Proceedings of the 2nd International Workshop on Science of Smart
  City Operations and Platforms Engineering}, ser. SCOPE '17.\hskip 1em plus
  0.5em minus 0.4em\relax Pittsburgh, Pennsylvania: ACM, 2017, pp. 7--12.

\bibitem{Kleberger2011}
P.~Kleberger, T.~Olovsson, and E.~Jonsson, ``Security aspects of the in-vehicle
  network in the connected car,'' in \emph{Proc. of IEEE Intelligent Vehicles
  Symposium (IV)}, Baden-Baden, Germany, June 2011, pp. 528--533.

\bibitem{Woo2015}
S.~Woo, H.~J. Jo, and D.~H. Lee, ``A practical wireless attack on the connected
  car and security protocol for in-vehicle can,'' \emph{IEEE Transactions on
  Intelligent Transportation Systems}, vol.~16, no.~2, April 2015.

\bibitem{Chaudhry2012}
H.~Chaudhry and T.~Bohn, ``Security concerns of a plug-in vehicle,'' in
  \emph{Proceedings of IEEE PES Innovative Smart Grid Technologies (ISGT)},
  Washington, DC, USA, Jan 2012, pp. 1--6.

\bibitem{Studnia2013}
I.~Studnia, V.~Nicomette, E.~Alata, Y.~Deswarte, M.~Kaâniche, and
  Y.~Laarouchi, ``Survey on security threats and protection mechanisms in
  embedded automotive networks,'' in \emph{Proc. of IEEE/IFIP Conference on
  Dependable Systems and Networks Workshop (DSN-W)}, Budapest, Hungary, June
  2013, pp. 1--12.

\bibitem{Calandriello2011}
G.~Calandriello, P.~Papadimitratos, J.~P. Hubaux, and A.~Lioy, ``On the
  performance of secure vehicular communication systems,'' \emph{IEEE
  Transactions on Dependable and Secure Computing}, vol.~8, no.~6, pp.
  898--912, Nov 2011.

\bibitem{Kim2010}
T.~Kim, A.~Studer, R.~Dubey, X.~Zhang, A.~Perrig, F.~Bai, B.~Bellur, and
  A.~Iyer, ``Vanet alert endorsement using multi-source filters,'' in
  \emph{Proceedings of the Seventh ACM International Workshop on VehiculAr
  InterNETworking}, Chicago, IL, USA, September 2010, pp. 51--60.

\bibitem{Sun2017}
M.~Sun, M.~Li, and R.~Gerdes, ``A data trust framework for vanets enabling
  false data detection and secure vehicle tracking,'' in \emph{Proc. of IEEE
  Conference on Communications and Network Security (CNS)}, Las Vegas, NV, USA,
  Oct 2017, pp. 1--9.

\bibitem{PETRILLO2018}
A.~Petrillo, A.~Pescapé, and S.~Santini, ``A collaborative approach for
  improving the security of vehicular scenarios: The case of platooning,''
  \emph{Computer Communications}, vol. 122, pp. 59 -- 75, 2018.

\bibitem{Tuohy2015}
S.~Tuohy, M.~Glavin, C.~Hughes, E.~Jones, M.~Trivedi, and L.~Kilmartin,
  ``Intra-vehicle networks: A review,'' \emph{IEEE Transactions on Intelligent
  Transportation Systems}, vol.~16, no.~2, pp. 534--545, April 2015.

\bibitem{BRACKSTONE1999181}
M.~Brackstone and M.~McDonald, ``Car-following: a historical review,''
  \emph{Transportation Research Part F: Traffic Psychology and Behaviour},
  vol.~2, no.~4, pp. 181 -- 196, 1999.

\bibitem{dorf2011modern}
R.~C. Dorf and R.~H. Bishop, \emph{Modern control systems}.\hskip 1em plus
  0.5em minus 0.4em\relax Pearson, 2011.

\bibitem{han2012game}
Z.~Han, D.~Niyato, W.~Saad, T.~Ba{\c{s}}ar, and A.~Hj{\o}rungnes, \emph{Game
  theory in wireless and communication networks: theory, models, and
  applications}.\hskip 1em plus 0.5em minus 0.4em\relax Cambridge University
  Press, 2012.

\bibitem{bacsar1998dynamic}
T.~Ba{\c{s}}ar and G.~J. Olsder, \emph{Dynamic noncooperative game
  theory}.\hskip 1em plus 0.5em minus 0.4em\relax SIAM, 1998.

\bibitem{heinrich2016deep}
J.~Heinrich and D.~Silver, ``Deep reinforcement learning from self-play in
  imperfect-information games,'' \emph{arXiv preprint arXiv:1603.01121}, 2016.

\bibitem{ferdowsi2018deep}
A.~Ferdowsi and W.~Saad, ``Deep learning for signal authentication and security
  in massive {Internet of Things} systems,'' \emph{arXiv preprint
  arXiv:1803.00916}, 2018.

\bibitem{chen2017machine}
M.~Chen, U.~Challita, W.~Saad, C.~Yin, and M.~Debbah, ``Machine learning for
  wireless networks with artificial intelligence: A tutorial on neural
  networks,'' \emph{arXiv preprint arXiv:1710.02913}, 2017.

\bibitem{LSTM}
A.~Graves, A.~R. Mohamed, and G.~Hinton, ``Speech recognition with deep
  recurrent neural networks,'' in \emph{Proc. of IEEE International Conference
  on Acoustics, Speech and Signal Processing}, May 2013.

\end{thebibliography}
	
	


\end{document}